\documentclass[reqno]{amsproc}
\usepackage{tikz}
\usepackage{amsmath}
\usepackage{amsfonts}
\usepackage{amssymb}
\usepackage{mathtools}
\usepackage{mathrsfs}
\usepackage[utf8x]{inputenc}
\usepackage[loose]{subfigure}
\usetikzlibrary{arrows,calc}
\usetikzlibrary{decorations.markings}
\usepackage{geometry}

\usepackage[all]{xy}

\definecolor{myurlcolor}{rgb}{0,0,0.4}
\definecolor{mycitecolor}{rgb}{0,0.5,0}
\definecolor{myrefcolor}{rgb}{0.5,0,0}
\usepackage{hyperref}
\hypersetup{colorlinks,
linkcolor=myrefcolor,
citecolor=mycitecolor,
urlcolor=myurlcolor}

\newcommand{\lt}{<}
\newcommand{\Z}{\mathbb{Z}}
\newcommand{\N}{\mathbb{N}}
\newcommand{\Q}{\mathbb{Q}}
\newcommand{\R}{\mathbb{R}}
\newcommand{\ra}{\rightarrow}
\newcommand{\eq}[1]{~(\ref{#1})}
\renewcommand{\equiv}{\stackrel{\mathrm{def}}{=}}
\newcommand{\eps}{\varepsilon}

\theoremstyle{plain}
\newtheorem{thm}{Theorem}
\newtheorem{lem}[thm]{Lemma}
\newtheorem{prop}[thm]{Proposition}
\newtheorem{cor}[thm]{Corollary}
\newtheorem{qstn}[thm]{Question}

\newtheorem{defn}[thm]{Definition}

\theoremstyle{definition}
\newtheorem{ex}[thm]{Example}

\theoremstyle{remark}
\newtheorem{rem}[thm]{Remark}

\newtheorem{prob}[thm]{Problem}

\newcommand{\beq}{\begin{equation}}
\newcommand{\eeq}{\end{equation}}

\begin{document}



\title[Velocity Polytopes]{Velocity Polytopes of Periodic Graphs\\[.1cm] and a No-Go Theorem for Digital Physics}

\author{}
\address{Tobias Fritz\\ Perimeter Institute for Theoretical Physics}
\email{tfritz@perimeterinstitute.ca}



\thanks{\textit{Acknowledgements.} This work grew out of a discussion with Giacomo Mauro D'Ariano. Thanks also to the organizers and participants of \emph{Experimental search for quantum gravity: the hard facts}, where this work has been discussed. Furthermore, an anonymous referee has provided very valuable feedback. Research at Perimeter Institute is supported by the Government of Canada through Industry Canada and by the Province of Ontario through the Ministry of Economic Development and Innovation. The author was previously supported by the EU STREP QCS}

\begin{abstract}
A periodic graph in dimension $d$ is a directed graph with a free action of $\Z^d$ with only finitely many orbits. It can conveniently be represented in terms of an associated finite graph with weights in $\Z^d$, corresponding to a $\Z^d$-bundle with connection. Here we use the weight sums along cycles in this associated graph to construct a certain polytope in $\R^d$, which we regard as a geometrical invariant associated to the periodic graph. It is the unit ball of a norm on $\R^d$ describing the large-scale geometry of the graph. It has a physical interpretation as the set of attainable velocities of a particle on the graph which can hop along one edge per timestep. Since a polytope necessarily has distinguished directions, there is no periodic graph for which this velocity set is isotropic. In the context of classical physics, this can be viewed as a no-go theorem for the emergence of an isotropic space from a discrete structure.
\end{abstract}

\maketitle

\subsection*{Corrigendum.}

As was kindly pointed out to the author by Toshikazu Sunada and Davide Proserpio, the ``velocity polytopes'' introduced in this paper are precisely the ``cycle figures'' of Eon~\cite{Eon2} and have also appeared in earlier works of Kotani and Sunada on random walks on periodic graphs~\cite{KS,KS2}; in particular, they coincide with the polytope $\mathcal{D}$ of~\cite{KS2} as characterized by Proposition 1.1 and Theorem 1.2, which imply our Proposition~\ref{lsgeometry}. Kotani and Sunada also make use of the language of homology theory~\cite{Hatcher}, which seems more appropriate than our terminology.

\section{Introduction}
\label{introduction}

Periodic graphs are abstractions of the atomic structure of crystals. A crystal, by definition, is a material whose structure consists of a finite-size pattern which repeats periodically in all spatial directions. Taking the crystal atoms as the vertices of a graph and the chemical bonds as its edges, one obtains a graph which repeats periodically in all spatial directions: a \emph{periodic graph}. This graph represents the chemical structure of the crystal. Therefore, the problem of classifying and enumerating all periodic graphs in three dimensions is of fundamental importance for crystallography~\cite{BP,DK,Eon,Thimm}. Periodic graphs have also been studied in operations research~\cite{HW}, spectral graph theory~\cite{Collatz}, and computer science~\cite{CM}. Certain generalizations of periodic graphs also appear in topological graph theory~\cite{GT,GTbook,Zas}.

Besides their natural appearance in all these fields, structures akin to periodic graphs have recently also been suggested as candidates for the fundamental microscopic building blocks of space. They feature prominently in \emph{quantum graphity}~\cite{Qgraphity}, related approaches to fundamental physics based on condensed matter models~\cite{BHmodel}, and in recent attempts to describe our world as a computation on a quantum computer~\cite{A1,A2}. In order to get some idea of how physics on a periodic graph can look like, let us consider a classical point particle moving on a periodic graph $\Gamma\subseteq\R^d$ as follows: the particle moves along the vertices of $\Gamma$ in discrete timesteps by hopping along one edge per timestep. More precisely, we define a \emph{trajectory} to be a sequence $(f_n)_{n\in\N}$ of vertices $f_n\in \Gamma$ such that for each time $n\in\N$, the positions $f_n$ and $f_{n+1}$ are adjacent in $\Gamma$. A vector $u\in\R^d$ is then called a \emph{velocity vector} of $\Gamma$ if there is a trajectory $f$ such that
\beq
\label{v}
u=\lim_{n \to \infty} \frac{f_n-f_0}{n} \:.
\eeq
Intuitively, this equation means that the trajectory's apparent velocity on the macroscopic scale is given by $u$. It is the trajectory's velocity as seen by a macroscopic observer who is not aware of the fundamental discreteness of $\Gamma$ and perceives space as a continuum $\R^d$. Note that\eq{v} only makes sense when the limit on the right-hand side exists, which can be interpreted as requiring the trajectory to have a well-defined constant macroscopic velocity, as required by Newton's first law. We will make~(\ref{v}) precise in a way which does not require the graph to be embedded in Euclidean space; our notion of velocity is completely abstract and combinatorial, but nevertheless accurately represents the usual concept. 

Now a natural question is: given the periodic graph $\Gamma\in\R^d$, what is the set of its velocity vectors? In particular, can this set be a Euclidean ball, thereby making the macroscopic observer perceive an isotropic space, such that the achievable absolute values $||u||$ do not depend on the direction $u/||u||$? This would be a very desirable property for the kind of models discussed e.g. in~\cite{A1,A2}.

Using concepts from the theory of periodic graphs, we will prove in Theorem~\ref{mainthm} that the set of velocity vectors of any suitably connected periodic graph $\Gamma$ is a convex polytope in $\R^d$. In particular, it never is a Euclidean ball, and the set of achievable velocity vectors cannot be isotropic; see Theorem~\ref{nogo}. This is a no-go theorem for the emergence of an isotropic space from a discrete structure within the context of classical physics. It extends the \emph{tile argument} of Weyl~\cite{Weyl,VB}. The reader only interested in this physics aspect may directly proceed to Section~\ref{physics}.

From the mathematical point of view, our \emph{velocity polytopes} are new invariants of periodic graphs. As witnessed by Proposition~\ref{lsgeometry}, they encode the periodic graph's large-scale geometry. The velocity polytope as an invariant can be applied for example as in Corollary~\ref{morphapp}, which is a criterion for proving the non-existence of translation-invariant maps between periodic graphs (``morphisms'').

\section{Preliminaries}

In this section, we collect some definitions and simple observations. Although most of the relevant literature is concerned with the case of undirected graphs~\cite{BP,CHK,CM}, we work with directed graphs, which is more general and has turned out to be technically more convenient. 

If $A$ is a finite set, we write $|A|$ for its cardinality.

\subsection{Graphs and paths.}
For us, a graph is a directed graph which may have loops and multiple edges. A graph $G$ is specified by a vertex set $V_G$, an edge set $E_G$, a source function $s_G:E_G\ra V_G$, and a target function $t_G:E_G\ra V_G$. We refrain from identifying an edge $e$ with the vertex pair $(s_G(e),t_G(e))$, since there may be several edges between $s_G(e)$ and $t_G(e)$. When the graph $G$ is clear from the context, we frequently omit the subscripts and simply write $s,t:E\ra V$ for the source and target maps in order to avoid unnecessary cluttering.

A \emph{path} $p$ in $G$ is a finite sequence of edges $p=e_1\ldots e_n$, $e_i\in E$, such that $t(e_i)=s(e_{i+1})$ for all $i=1,\ldots,n-1$. The \emph{length} $|p|=n$ of $p$ is its number of edges. The empty path $\emptyset$ is the unique path of length $0$. 

A \emph{closed path} is a non-empty path $p=e_1\ldots e_n$ such that $s(e_1)=t(e_n)$. If no other additional vertex repetitions occur, then $p$ is also called a \emph{cycle}. By this definition, a cycle of length $n$ traverses $n$ distinct vertices. 

If $p=e_1\ldots e_n$ and $p'=e'_1\ldots e'_m$ are paths such that $t(e_n)=s(e'_1)$, then $p$ and $p'$ can be composed to a new path
$$
pp'\equiv e_1\ldots e_n e'_1\ldots e'_m \:.
$$
It is clear that $|pp'| = |p| + |p'|$.

A directed graph is said to be \emph{strongly connected} if there is a path from $v$ to $w$ for any two vertices $v,w\in V$.

\begin{lem}
\label{fingraph}
Let $G$ be a finite graph.
\begin{enumerate}
\item\label{pathcycle} A path $p=e_1\ldots e_n$ in $G$ of length $n\geq |V_G|$ contains a cycle: there are indices $k$ and $l$ such that $c=e_{k}e_{k+1}\ldots e_{l}$ is a cycle.
\item\label{fincycles} There are only a finite number of cycles in $G$.
\end{enumerate}
\end{lem}

\begin{proof}
\begin{enumerate}
\item By the pigeonhole principle, there have to be indices $l>k$ with $t(e_l)=s(e_k)$. Choosing $l$ minimal with this property guarantees $c=e_{k}e_{k+1}\ldots e_{l}$ to be a cycle.
\item Since a cycle is defined as not having any vertex repetitions besides the coincidence between the initial and the final vertex, a cycle can have length at most $|V_G|$. The conclusion follows since there are only a finite number of paths of length at most $|V_G|$.
\end{enumerate}
\end{proof}

\subsection{Definition of periodic graphs.}
\label{depper}

We now turn to the formal definition of periodic graphs before discussing their representation by finite weighted graphs. We refer to Figure~\ref{figure} for a basic two-dimensional example and to~\cite{BP} for abundant visualizations of three-dimensional periodic graphs within the context of crsytallography.

\begin{defn}[periodic graph]
\label{pg}
A $d$-dimensional periodic graph is a graph $\Gamma$ equipped with a free action of the free abelian group $\Z^d$ on $\Gamma$, such that $\Gamma$ has only finitely many $\Z^d$-orbits of vertices as well as edges.
\end{defn}

Let us disentangle what this definition means. First of all, the graph $\Gamma$ comes with an action of the group $\Z^d$. In additive notation, this means that there are given maps\\[-.2cm]
\begin{align*}
\begin{split}
V_\Gamma &\times\Z^d \rightarrow V_\Gamma , \quad (v,x)\mapsto v+x \:, \\[.2cm]
E_\Gamma &\times\Z^d \rightarrow E_\Gamma , \quad (e,x)\mapsto e+x \:. \\[.2cm]
\end{split}
\end{align*}
We think of the vertex $v+x$ as the vertex $v$ translated by the vector $x\in\Z^d$, and similarly for $e+x$. In order for these maps to form a $\Z^d$-action on $\Gamma$, they need to satisfy the group action axioms \\[-.2cm]
\begin{align}
\label{action}
\begin{split}
 (v+x)+y=v+(x+y) \quad\forall v\in V_\Gamma ,\: x,y\in\Z^d \:, \qquad v+0 &= v \quad\forall v\in V_\Gamma  \:, \\[.2cm]
 (e+x)+y=e+(x+y) \quad\forall e\in E_\Gamma ,\: x,y\in\Z^d \:, \qquad e+0 &= e \quad\forall e\in E_\Gamma  \:, \\[.2cm]
\end{split}
\end{align}
as well as be compatible with each other in the sense that source and target of the translate of an edge are precisely the translates of the source and target of the edge,\\[-.2cm]
\begin{align}
\label{comp}
\begin{split}
 & s_\Gamma (e+x)=s_\Gamma (e)+x\:, \\[.2cm]
 & t_\Gamma (e+x)=t_\Gamma (e)+x\:. \\[.2cm]
\end{split}
\end{align}
Furthermore, the action of $\Z^d$ on $\Gamma$ should be free,
\beq
\label{free}
v+x=v+x' \:\:\Longrightarrow\:\: x=x' \:\:,\qquad e+x=e+x' \:\:\Longrightarrow\:\: x=x' \:.
\eeq
Finally, there should be only finitely many orbits in $V_\Gamma $ as well as in $E_\Gamma $ under the $\Z^d$-action. This is an abstraction of the crystallographic property that a unit cell of a crystal contains only finitely atoms and chemical bonds.

\begin{rem}
\label{affine}
Definition~\ref{pg} is of an abstract combinatorial nature in the sense that no embedding of $\Gamma$ into $\R^d$ is required. However, suppose that $\Gamma\subseteq\R^d$ is a concretely embedded graph which is translation-invariant in $d$ linearly independent directions. After applying an appropriate affine transformation to $\Gamma$, the unit cell of $\Gamma$ can be taken to be the unit cube $[0,1]^d$, which means that $\Gamma$ is translation-invariant under the group of integer translations $\Z^d\subseteq\R^d$. With $\Z^d$ acting on $\Gamma$ by these translations, $\Gamma$ is a $d$-dimensional periodic graph in the sense of Definition~\ref{pg}.
\end{rem}

There are other uses of the term ``periodic graph'' in the mathematical literature which are not related to the one used here. For example, a plot of a periodic function is a ``periodic graph'' in a completely different sense. For another interesting notion of periodic graph which is in no way related to the present one, see~\cite{Godsil}.

\subsection{The displacement graph of a periodic graph.}
\label{dissec}

By the freeness condition\eq{free}, a non-empty periodic graph is necessarily infinite. A convenient representation of a periodic graph in terms of a finite amount of data has been developed in a more general context in~\cite{Gross,GT}, and probably independently in~\cite{CHK}. Due to the diversity of the literature spanning various fields of science, no universal terminology has been established. Here we partly try to follow the terminology of graph theory~\cite{GTbook}. While this subsection contains standard material, we try to offer a slightly different point of view emphasizing the analogy to covering spaces~\cite[Ch.~1.3]{Hatcher}.

Given a group action on some mathematical object, it is natural to consider the quotient object with respect to the group action. For a periodic graph $\Gamma$, this means to identify two vertices (or edges) if they can be translated into each other by a group element $x\in\Z^d$; in other words, if the two vertices (edges) lie in the same $\Z^d$-orbit. The resulting collection of vertex orbits $V_{\Gamma/\Z^d}\equiv V_\Gamma/\Z^d$ and the collection of edge orbits $E_{\Gamma/\Z^d}\equiv E_\Gamma/\Z^d$ form a quotient graph $\Gamma/\Z^d$: by\eq{comp}, the source and target functions $s_\Gamma$ and $t_\Gamma$ descend to well-defined maps
$$
s_{\Gamma/\Z^d},\, t_{\Gamma/\Z^d} \: : \: E_{\Gamma/\Z^d}\longrightarrow V_{\Gamma/\Z^d} \:.
$$
By the finiteness assumption of Definition~\ref{pg}, the quotient graph $\Gamma/\Z^d$ is finite.

\begin{rem}
In the particular case that $\Gamma\subseteq\R^d$ is a translation-invariant Euclidean periodic graph, then one can construct $\Gamma/\Z^d$ also by taking the vertices and edges in a unit cell of $\Gamma$. Besides the edges inside the unit cell, each edge of $\Gamma$ which connects a vertex inside the unit cell to a vertex outside the unit cell defines an additional edge of $\Gamma/\Z^d$ by changing the target vertex to its translate inside the unit cell.
\end{rem}

By definition of $\Gamma/\Z^d$, there is a canonical projection map $\phi_\Gamma:\Gamma\ra \Gamma/\Z^d$ which maps every vertex and every edge to its $\Z^d$-orbit. When $\Gamma$ is clear from the context, we also simply write $\phi$ for $\phi_\Gamma$.

The map $\phi$ enjoys the nice property that an edge (or a path) in $\Gamma/\Z^d$ can be uniquely lifted to an edge (a path) in $\Gamma$, given that a starting vertex has been specified:

\begin{lem}
\label{lift}
\begin{enumerate}
\item\label{liftedge} For every $e\in E_{\Gamma/\Z^d}$ and every $v^*\in \phi^{-1}(s(e))$, there is a unique $e^*\in E_\Gamma$ with $\phi(e^*)=e$ and $s(e^*)=v^*$.
\item\label{liftpath} For every path $p=e_1\ldots e_n$ in $\Gamma/\Z^d$ and every $v^*\in \phi^{-1}(s(e_1))$, there is a unique path $p^*=e_1^*\ldots e_n^*$ in $\Gamma$ with $\phi(p^*)=p$ and $s(e_1^*)=v^*$.
\end{enumerate}
\end{lem}

\begin{proof}
\begin{enumerate}
\item Let $\widetilde{e}\in E_\Gamma$ be some edge with $\phi(\widetilde{e})=e$. Then there is a unique $x\in\Z^d$ such that $s(\widetilde{e})+x=v^*$. Hence, $e^*\equiv \widetilde{e}+x$ has the desired properties. For uniqueness, suppose that $e'\in E_\Gamma$ would also satisfy $\phi(e')=e$ and $s(e')=v^*$. By definition of $\phi$, the relation $\phi(e^*)=e=\phi(e')$ means that $e^*$ and $e'$ lie in the same $\Z^d$-orbit, so that there exists $y\in\Z^d$ with $e'=e^*+y$. But then, $v^*=s(e')=s(e^*)+y=v^*+y$, which implies $y=0$ by\eq{free}. Hence $e'=e^*$.
\item This follows from a successive application of part~\ref{liftedge} to each edge in the path.
\end{enumerate}
\end{proof}

In the language of graph theory, we have found that the projection $\phi:\Gamma\ra \Gamma/\Z^d$ is a \emph{covering of graphs}~\cite[Ch.\,2]{GTbook},~\cite[Ch.\,17]{Biggs}. This is completely analogous to the notion of \emph{covering space} in topology~\cite[Ch.\,1.3]{Hatcher}.

Unfortunately, knowing the quotient graph $\Gamma/\Z^d$ alone is not enough to reconstruct $\Gamma$. For example, there are many Euclidean periodic graphs $\Gamma\subseteq\R^d$ which contain only a single vertex per unit cell, so that $|V_{\Gamma/\Z^d}|=1$. In this case, all edges of $\Gamma/\Z^d$ are loops. However, knowing $|E_{\Gamma/\Z^d}|$ as the number of these loops is certainly not enough to recover $\Gamma$: for example, it is unclear whether a loop of $\Gamma/\Z^d$ comes from an orbit of loops in $\Gamma$, or whether it represents a class of edges connecting different vertices in $\Gamma$.

The additional piece of data needed in order to recover $\Gamma$ turns out to consist of edge weights on $\Gamma/\Z^d$ with values in $\Z^d$, which specify, intuitively speaking, the translation required in going from $s_\Gamma(e)$ to $t_\Gamma(e)$. These edge weights are known under various names --- e.g. \emph{voltage assignments}~\cite{Gross,GT}, \emph{labels}~\cite{CHK,Klee}, or simply \emph{weights}~\cite{CM}. We will prefer the term \emph{displacements}, which we deem most appropriate given the geometric intuition. The formalism described in the following is, in effect, gauge theory for the group $\Z^d$~\cite{gauge}.

Defining displacements requires that one has chosen a vertex representative $\iota(v)\in V_\Gamma$ for every orbit $v\in V_{\Gamma/\Z^d}$. In other words, we fix a map $\iota:V_{\Gamma/\Z^d}\to V_\Gamma$ which is assumed to be a section of $\phi:V_\Gamma\to V_{\Gamma/\Z^d}$. For example for a Euclidean periodic graph $\Gamma\subseteq\R^d$, one possibility is to define $\iota$ by choosing a unit cell and mapping every orbit $v\in V_{\Gamma/\Z^d}$ to its representative in the unit cell. 

Every $v\in V_{\Gamma/\Z^d}$ represents a whole $\Z^d$ worth of vertices of $\Gamma$, namely $\phi^{-1}(v)$. Defining $\iota(v)$ means fixing an origin in $\phi^{-1}(v)$, in the sense that it gives the concrete identification of this $\phi^{-1}(v)$ with $\Z^d$ via
$$
\alpha_v \::\: \Z^d \stackrel{\cong}{\longrightarrow} \phi^{-1}(v) \:,\quad x\mapsto \iota(v)+x \:.
$$
This map is compatible with translations in the sense that it satisfies the identity $\alpha_v(x+y)=\alpha_v(x)+y$ for all $x,y\in\Z^d$.

Now for an edge orbit $e\in E_{\Gamma/\Z^d}$, Lemma~\ref{lift} provides a bijection $\phi^{-1}(s(e))\stackrel{\cong}{\ra}\phi^{-1}(t(e))$, which is also compatible with translations. In total, we obtain an isomorphism
\beq
\label{defgamma}
\xymatrix{\gamma_e  \:\:: \:\: \Z^d \ar[r]^{\cong}_{\alpha_{s(e)}} & \phi^{-1}(s(e)) \ar[r]^{\cong} & \phi^{-1}(t(e)) \ar[r]^(.6){\cong}_(.6){\alpha_{t(e)}^{-1}} & \Z^d}
\eeq
which is again compatible with translations, $\gamma_e(x+y)=\gamma_e(x)+y$. Therefore, $\gamma_e(x)=\gamma_e(0)+x$. We now define the displacement along $e$ to be $\delta(e)\equiv\gamma_e(0)\in\Z^d$. The equation $\gamma_e(0)=\delta(e)$ expresses the intuition that $\delta(e)$ is the physical displacement required in going from $s_\Gamma(e')$ to $t_\Gamma(e')$ for any $e'\in\phi^{-1}(e)$.

\begin{rem}
In gauge theory terms, $\phi:\Gamma\to\Gamma/\Z^d$ is a $\Z^d$-\emph{principal bundle}, the section $\iota:\Gamma/\Z^d\to\Gamma$ fixes a \emph{trivialization}, and the map $\delta : \Gamma/\Z^d\to \Z^d$ defines a $\Z^d$-\emph{connection} on $\Gamma/\Z^d$.
\end{rem}

It is not difficult to see that the quotient graph $\Gamma/\Z^d$ together with the displacement function $\delta:E_{\Gamma/\Z^d}\ra\Z^d$ is sufficient to recover $\Gamma$. This works as in the following definition, which can also be regarded as a general scheme for constructing periodic graphs:

\begin{defn}[{\cite{GT,CHK,CM}}]
\label{depgraph}
A \emph{displacement graph} $(G,\delta)$ is a finite graph $G$ together with edge weights $\delta:E_G \to \Z^d$ (the \emph{displacements}). Associated to $(G,\delta)$ is a periodic graph $\widetilde{G}$ given by \\[-.2cm]
\begin{align*}
V_{\widetilde{G}} =&\left\{(v,x) \:|\: v\in V_G,\, x\in\Z^d \right\} \\[.2cm]
E_{\widetilde{G}} =&\left\{(e,x) \:|\: e\in E_G,\, x\in\Z^d \right\} \\[.4cm]
s_{\widetilde{G}} ((e,x)) =  (s_G(e)&, x) \:,\qquad t_{\widetilde{G}} ((e,x)) = (t_G(e),x+\delta(e))
\end{align*}
\end{defn}

\begin{rem}
\label{int}
Intuitively, $\widetilde{G}$ is constructed from $G$ as follows: we start with the lattice $\Z^d$ and place at each point a copy of the vertex set $V_G$. Each edge $e\in E_G$ defines a $\Z^d$ worth of edges between the copies of $s_G(e)$ and $t_G(e)$, where in adding these edges we have to apply a translation by $\delta(e)$ in the ambient $\Z^d$.
\end{rem}

One needs to keep in mind that determining a displacement graph from a periodic graph $\Gamma$ requires choosing a representative $\iota(v)\in V_\Gamma$ for each $\Z^d$-orbit $v\in V_{\Gamma/\Z^d}$. How much does the displacement graph depend on the choice of $\iota$? Suppose we are given two different choices $\iota,\iota':V_{\Gamma/\Z^d}\to V_\Gamma$. Then for any $v\in V_{\Gamma/\Z^d}$, there is a unique $g(v)\in\Z^d$ such that
$$
\iota'(v) = \iota(v)+g(v) \quad \forall v\in V_{\Gamma/\Z^d} \:.
$$
By\eq{defgamma}, the corresponding displacements $\delta$ and $\delta'$ therefore differ by
$$
\delta'(e) = \delta(e) + g(s(e)) - g(t(e)) \quad \forall e\in E_{\Gamma/\Z^d} \:.
$$

\begin{rem}
In gauge theory terms, this equation corresponds to conducting a \emph{gauge transformation}.
\end{rem}

\begin{ex}
\label{hc1}
Consider the periodic graph illustrated in Figure~\ref{honeycomb}. Any one of the elementary parallelograms formed by the dashed lines can be taken as a unit cell. Choosing the two vertices inside such a unit cell defines $\iota$ in terms of a representative of the set of ``\begin{tikzpicture}\node[circle,inner sep=2pt,fill=black,draw]{};\end{tikzpicture}'' vertices which form a $\Z^d$-orbit, and a representative of the set of ``\begin{tikzpicture}\node[rectangle,inner sep=2pt,fill=black,draw]{};\end{tikzpicture}'' vertices which form another $\Z^d$-orbit. The associated displacement graph is shown in Figure~\ref{hcdg}. One obtains the displacements of e.g.~the edges going from \begin{tikzpicture}\node[circle,inner sep=2pt,fill=black,draw]{};\end{tikzpicture} to \begin{tikzpicture}\node[rectangle,inner sep=2pt,fill=black,draw]{};\end{tikzpicture} by noting that there are three $\Z^d$-orbits of such edges in Figure~\ref{honeycomb}: one which stays inside the unit cell, corresponding to the displacement $(0,0)$; one whose target is one cell away in the positive $y$-direction, having displacement $(0,1)$; and one whose target is one cell away in the negative $x$-direction with displacement $(-1,0)$.

If one started instead with the displacement graph~\ref{hcdg}, one would probably draw its associated periodic graph as in Figure~\ref{hcrect}, which is a different Euclidean embedding of the same periodic graph as in Figure~\ref{honeycomb}.
\end{ex}


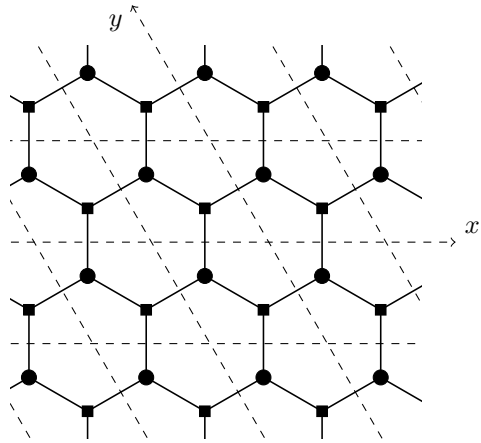
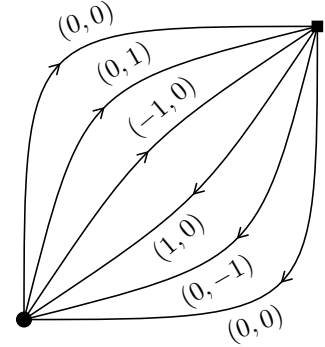
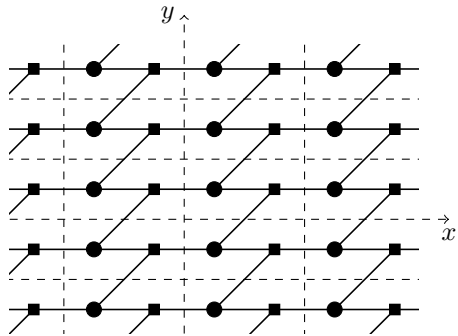
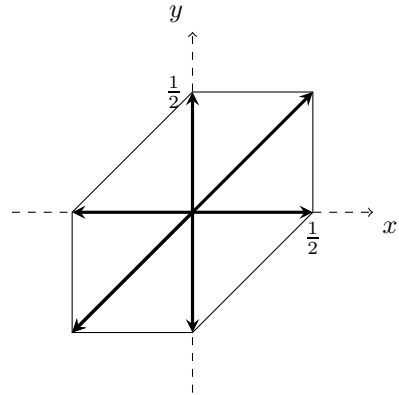
\begin{figure}
\begin{center}
\subfigure[A periodic graph periodically embedded in $\R^2$. Every edge represents a parallel pair of edges with opposite orientation.]{\label{honeycomb}
\begin{tikzpicture}[scale=0.9]
\newcommand{\latticesize}{7}
\newcommand{\xaxis}{2}
\newcommand{\yaxis}{-4}
\begin{scope}
\clip (0.6,0.6) rectangle (\latticesize+1.4-1.732,\latticesize-.6) ;
\tikzstyle{vertex}=[inner sep=2pt,fill=black,draw]
\tikzstyle{edges}=[semithick]
\foreach \x in {-\latticesize,...,\latticesize}
{
	\foreach \y in {-\latticesize,...,\latticesize}
	{
		\node at (-1.732*\x - 0.866*\y, 1.5*\y) [circle,vertex] {} ;
		\node at (-1.732*\x - 0.866*\y, 1.5*\y + 1) [rectangle,vertex] {} ;
		\draw [edges] (-1.732*\x - 0.866*\y, 1.5*\y) -- (-1.732*\x - 0.866*\y, 1.5*\y + 1) ;	
		\draw [edges] (-1.732*\x - 0.866*\y, 1.5*\y) -- (-1.732*\x - 0.866*\y - 0.866, 1.5*\y - 0.5) ;	
		\draw [edges] (-1.732*\x - 0.866*\y, 1.5*\y) -- (-1.732*\x - 0.866*\y + 0.866, 1.5*\y - 0.5) ;	
	}
}
\foreach \x in {-\latticesize,...,\yaxis}
	\draw [dashed] (-1.732*\x+1.732 - .5, 0) -- (-1.732*\x+1.732 - 0.866*\latticesize - .5, 1.5*\latticesize) ;
\foreach \x in {\yaxis,...,\latticesize}
	\draw [dashed] (-1.732*\x-1.732 - .5, 0) -- (-1.732*\x-1.732 - 0.866*\latticesize - .5, 1.5*\latticesize) ;
\foreach \y in {0,...,\xaxis}
	\draw [dashed] (1.732*\latticesize - 0.866*\y+0.866, 1.5*\y-1.5 + 0.5) -- (-1.732*\latticesize - 0.866*\y+0.866, 1.5*\y-1.5 + 0.5) ;
\foreach \y in {\xaxis,...,\latticesize}
	\draw [dashed] (1.732*\latticesize - 0.866*\y-0.866, 1.5*\y+1.5 + 0.5) -- (-1.732*\latticesize - 0.866*\y-0.866, 1.5*\y+1.5 + 0.5) ;
\end{scope}
\draw [->,dashed] (-1.732*\yaxis - 0.5 - 0.577*0.6, 0.6) -- (-1.732*\yaxis - 0.577*\latticesize - .5, \latticesize) node [anchor=north east] {$y$} ;
\draw [->,dashed] (0.6, 1.5*\xaxis + 0.5) -- (\latticesize+1.9-1.732, 1.5*\xaxis + 0.5) node [anchor=south west] {$x$} ;
\end{tikzpicture}}\hspace{1cm}
\subfigure[The displacement graph associated to \subref{honeycomb} with weights $(x,y)$.]{\label{hcdg}
\begin{tikzpicture}[scale=.78]
\clip (-.5,-.5) rectangle (5.5,5.5) ;
\tikzstyle{vertex}=[inner sep=2pt,fill=black,draw]
\tikzstyle{edges}=[semithick]
\node at (0,0) [circle,vertex] {} ;
\node at (5,5) [rectangle,vertex] {} ;
\draw [edges,decoration={markings,mark=at position 0.5 with \arrow{angle 60}},postaction=decorate] (0,0) .. controls (0,5) .. (5,5) node [pos=.62,sloped,above] {$(0,0)$} ;
\draw [edges,decoration={markings,mark=at position 0.5 with \arrow{angle 60}},postaction=decorate] (0,0) .. controls (1,4) .. (5,5) node [pos=.62,sloped,above] {$(0,1)$} ;
\draw [edges,decoration={markings,mark=at position 0.5 with \arrow{angle 60}},postaction=decorate] (0,0) .. controls (2,3) .. (5,5) node [pos=.62,sloped,above] {$(-1,0)$} ;
\draw [edges,decoration={markings,mark=at position 0.5 with \arrow{angle 60}},postaction=decorate] (5,5) .. controls (3,2) .. (0,0) node [pos=.62,sloped,below] {$(1,0)$} ;
\draw [edges,decoration={markings,mark=at position 0.5 with \arrow{angle 60}},postaction=decorate] (5,5) .. controls (4,1) .. (0,0) node [pos=.62,sloped,below] {$(0,-1)$} ;
\draw [edges,decoration={markings,mark=at position 0.5 with \arrow{angle 60}},postaction=decorate] (5,5) .. controls (5,0) .. (0,0) node [pos=.62,sloped,below] {$(0,0)$} ;
\end{tikzpicture}}\\[.5cm]
\subfigure[A different embedding of \subref{honeycomb}.]{\label{hcrect}\begin{tikzpicture}[scale=.8]
\newcommand{\latticesize}{6}
\newcommand{\xaxis}{2}
\newcommand{\yaxis}{2}
\begin{scope}
\clip (0.6,0.6) rectangle (\latticesize+1.4,\latticesize-.6) ;
\tikzstyle{vertex}=[inner sep=2pt,fill=black,draw]
\tikzstyle{edges}=[semithick]
\foreach \x in {0,...,\latticesize}
{
	\foreach \y in {0,...,\latticesize} \node at (2*\x,\y) [circle,vertex] {} ;
	\foreach \y in {0,...,\latticesize}
	{
		\node at (2*\x+1,\y) [rectangle,vertex] {} ;
		\draw [edges] (2*\x,\y) -- (2*\x+1,\y) ;	
		\draw [edges] (2*\x-1,\y) -- (2*\x,\y) ;	
		\draw [edges] (2*\x,\y) -- (2*\x+1,\y+1) ;	
	}
}
\foreach \x in {1,...,\yaxis}
	\draw [dashed] (2*\x-2-.5,0) -- (2*\x-2-.5,\latticesize) ;
\foreach \x in {\yaxis,...,\latticesize}
	\draw [dashed] (2*\x+2-.5,0) -- (2*\x+2-.5,\latticesize) ;
\foreach \y in {0,...,\xaxis}
	\draw [dashed] (0,\y-1+.5) -- (\latticesize+2,\y-1+.5) ;
\foreach \y in {\xaxis,...,\latticesize}
	\draw [dashed] (0,\y+1+.5) -- (\latticesize+2,\y+1+.5) ;
\end{scope}
\draw [->,dashed] (2*\yaxis-.5,0.6) -- (2*\yaxis-.5,\latticesize-0.1) node [anchor=east] {$y$} ;
\draw [->,dashed] (0.6,\xaxis+.5) -- (\latticesize+1.9,\xaxis+.5) node [anchor=north] {$x$} ;
\end{tikzpicture}}\hspace{1cm}
\subfigure[The nonzero basic velocities (arrows) and the velocity polytope (hexagon).]{\label{hcvp}\begin{tikzpicture}[scale=.8]
\draw[->,dashed] (-3,0) -- (3,0) node [anchor=north west] {$x$} ;
\draw[->,dashed] (0,-3) -- (0,3) node [anchor=south east] {$y$} ;
\draw[very thick,->,>=stealth] (0,0) -- (2,0) node [anchor=north] {$\frac{1}{2}$} ;
\draw[very thick,->,>=stealth] (0,0) -- (-2,0) ;
\draw[very thick,->,>=stealth] (0,0) -- (0,2) node [anchor=east] {$\frac{1}{2}$} ;
\draw[very thick,->,>=stealth] (0,0) -- (0,-2) ;
\draw[very thick,->,>=stealth] (0,0) -- (2,2) ;
\draw[very thick,->,>=stealth] (0,0) -- (-2,-2) ;
\draw[very thick,->,>=stealth] (0,0) -- (-2,-2) ;
\draw (2,0) -- (2,2) -- (0,2) -- (-2,0) -- (-2,-2) -- (0,-2) -- (2,0) ;
\end{tikzpicture}}
\end{center}
\caption{A periodic graph in two different embeddings~\subref{honeycomb},~\subref{hcrect}, its associated displacement graph~\subref{hcdg} and its velocity polytope~\subref{hcvp}. For more detail, see Examples~\ref{hc1} and~\ref{hc2}.}
\label{figure}
\end{figure}

\section{Velocity Polytopes}
\label{main}

From now on, we assume $\Gamma$ to be a periodic graph equipped with a fixed choice of orbit representatives $\iota:V_{\Gamma/\Z^d}\to V_\Gamma$. If $p=e_1\ldots e_n$ is a path in $\Gamma/\Z^d$, then by abuse of notation we define its displacement to be given by
\beq
\label{deltap}
\delta(p)\equiv \sum_i \delta(e_i) \:,
\eeq
which is nicely coherent with the lifting properties of Lemma~\ref{lift}: if each $e_i$ lifts to an edge which intuitively translates by $\delta(e_i)$, then the path $e_1\ldots e_n$ should lift to a path which intuitively tranlates by $\sum_i \delta(e_i)$.

\subsection{Velocity.} We now formalize the concepts introduced in the introduction. For technical convenience, we formally define a trajectory as a sequence of edges rather than vertices:

\begin{defn}
\label{deftraj}
A \emph{trajectory} in $\Gamma$ is a sequence $(f_n)_{n\in\N}$ of edges $f_n\in E_\Gamma$ such that $s(f_{n+1})=t(f_n)$ for all $n\in\N$.
\end{defn}

Intuitively, a trajectory is nothing but an infinite path in $\Gamma$. Thanks to Lemma~\ref{lift}, up to an overall translation a trajectory in $\Gamma$ is uniquely specified by the sequence of edges $\phi(f_n)\in E_{\Gamma/\Z^d}$, which are the images under the projection $\phi:\Gamma\to \Gamma/\Z^d$. In the following, we will abuse notation by also writing $f_n$ for $\phi(f_n)$.

By the definition\eq{deltap}, the displacement traversed by the trajectory $f$ between $n=n_1$ and $n=n_2$, i.e. along the path $f_{n_1}\ldots f_{n_2-1}$, is given by
$$
\sum_{k=n_1}^{n_2-1} \delta(f_k) \:.
$$
Since this displacement gets traversed in $n_2-n_1$ timesteps, it makes sense to define the \emph{velocity} in that time interval to be given by the difference quotient
\beq
\label{finvel}
\frac{\sum_{k=n_1}^{n_2-1} \delta(f_k)}{n_2-n_1}
\eeq
The trajectory $f$ has a well-defined \emph{velocity} if the limit
\beq
\label{vel}
u_f=\lim_{n\to\infty} \frac{\sum_{k=1}^n \delta(f_k)}{n} \:\:\in\: \R^d
\eeq
exists. In this case, the difference quotient\eq{finvel} also converges to $u_f$ for $n_2\to\infty$ with any fixed $n_1\in\N$.

In the following, $||\cdot||$ will be a fixed but arbitrary norm on $\R^d$. 

\begin{lem}
\label{iotaind}
The velocity $u_f$ of the trajectory $f$ does not depend on the particular choice of representatives $\iota:V_{\Gamma/\Z^d}\to V_\Gamma$ used for constructing the displacement function $\delta$, but only on $\Gamma$ and $f$ themselves.
\end{lem}

\begin{proof}
Let $\iota,\iota':V_{\Gamma/\Z^d}\to V_\Gamma$ be two choices of orbit representatives. It has been noted in Section~\ref{dissec} that their displacement functions satisfy
$$
\delta'(e) = \delta(e) + g(s(e)) - g(t(e))
$$
for some appropriate function $g:V_{\Gamma/\Z^d}\to\Z^d$. Then the displacements associated to the path $f_1\ldots f_n$ differ by
\beq
\label{regauge}
\delta'(f_1\ldots f_n) = \delta(f_1\ldots f_n) + \sum_{k=1}^n \big[ g(s(f_k)) - g(t(f_k)) \big] \:.
\eeq
Due to $s(f_{k+1}) = t(f_k)$, the sum is telescoping, so that
$$
\delta'(f_1\ldots f_n) = \delta(f_1\ldots f_n) + g(s(f_1)) - g(t(f_n)) \:.
$$
Writing $C\equiv\max_{v\in V} ||g(v)||$, we conclude that
$$
\left|\left|\frac{\sum_{k=1}^n \delta'(f_k)}{n} - \frac{\sum_{k=1}^n\delta(f_k)}{n} \right|\right| \leq \frac{2C}{n}
$$
from which the assertion immediately follows by taking the limit $n\ra\infty$.
\end{proof}

Note that the velocity of a finite path as in\eq{finvel} is in general not well-defined in the sense of the lemma, but does depend on the choice of $\iota$.

\begin{rem}
\label{physics}
When $\Gamma\subseteq\R^d$ is a Euclidean periodic graph, then we would like this notion of velocity to correspond to the usual one familiar from classical mechanics. As already mentioned in Remark~\ref{affine}, we can always apply an affine transformation to $\Gamma$ such that the unit cell becomes the ordinary unit cube $[0,1]^d$. As a matter of bookkeeping, this also changes all velocity vectors by the same affine transformation.

We claim that if the unit cell of $\Gamma$ is the unit cube, then Definition\eq{vel} gives precisely the usual concept of velocity vector. To see this, let us choose the orbit representatives $\iota$ to be those in the unit cell, and write $M$ for the maximal distance between any two vertices in the unit cell. Then, the actual distance vector traversed along the path $f_1\ldots f_n$ will differ from the displacement $\sum_k \delta(f_k)$ by at most $2M$. In the limit as $n\to\infty$, this is negligible, since all distances get divided by the total elapsed time $n$. This proves the claim.
\end{rem}

As a trivial example, a constant trajectory has a velocity of $0$. Similarly for any trajectory which stays in a bounded region in $\Gamma\subseteq\R^d$.

More non-trivial examples of velocities for arbitrary $\Gamma$ are as follows: for a cycle $c=e_1\ldots e_n$ in $\Gamma/\Z^d$, we define the \emph{basic velocity} associated to $c$ to be given by
$$
u_c=\frac{\delta(c)}{n}=\frac{\sum_{k=1}^n \delta(e_k)}{n}
$$
This coincides with\eq{finvel}. Lemma~\ref{lift} guarantees that $c$ lifts to a unique path in $\Gamma$ upon choosing an arbitrary vertex in $\phi^{-1}(s(e_1))$ as starting point. The trajectory defined by lifting a periodic traversal of $c$ from $\Gamma/\Z^d$ to $\Gamma$ has the basic velocity $u_c$ as its velocity.

One can regard the set of basic velocities as an invariant of the periodic graph:

\begin{lem}
\label{gaugeinv}
For a given cycle $c$, the basic velocity $u_c$ does not depend on the particular choice of $\iota:V_{\Gamma/\Z^d}\to V_\Gamma$ used in constructing the displacement function $\delta$.
\end{lem}

\begin{proof}
Applying equation\eq{regauge} in this case, one finds that all terms in the sum cancel each other, so that $\delta'(c) = \delta(c)$.
\end{proof}

\begin{rem}
In gauge theory terms, a basic velocity is the \emph{curvature} of the connection $\delta$. Lemma~\ref{gaugeinv} states that curvature is invariant under gauge transformations.
\end{rem}

Lemma~\ref{fingraph}\ref{fincycles} implies that there is only a finite number of basic velocities for fixed $\Gamma$.

\begin{ex}
\label{hc2}
We go back to the periodic graph illustrated in Figure~\ref{figure}. There are $9$ cycles in the displacement graph of Figure~\ref{hcdg} which have ``\begin{tikzpicture}\node[circle,inner sep=2pt,fill=black,draw]{};\end{tikzpicture}'' as their starting vertex, all of length $2$; the other $9$ cycles with ``\begin{tikzpicture}\node[rectangle,inner sep=2pt,fill=black,draw]{};\end{tikzpicture}'' as their starting vertex have the same basic velocities, so it is sufficient to consider the former. The basic velocities are\\
\begin{align*}
\phantom{-}\frac{(0,0)+(0,0)}{2}\:,\qquad \phantom{-}\frac{(0,0)+(0,-1)}{2}\:,\qquad &\phantom{-}\frac{(0,0)+(1,0)}{2}\:,\\[.2cm]
\phantom{-}\frac{(0,1)+(0,0)}{2}\:,\qquad \phantom{-}\frac{(0,1)+(0,-1)}{2}\:,\qquad &\phantom{-}\frac{(0,1)+(1,0)}{2}\:,\\[.2cm]
\frac{(-1,0)+(0,0)}{2}\:,\qquad \frac{(-1,0)+(0,-1)}{2}\:,\qquad &\frac{(-1,0)+(1,0)}{2}\:.\\[-.2cm]
\end{align*}
The nonzero ones are depicted in Figure~\ref{hcvp}.
\end{ex}

Now that we have seen some examples of velocities, a natural question to ask is the following:

\begin{qstn}
\label{mainqstn}
Given a periodic graph $\Gamma$, what is the set of its velocities?
\end{qstn}

\subsection{Main theorem.} We now proceed to answer Question~\ref{mainqstn} and give an explicit description of the set of velocities of $\Gamma$ as a subset of $\R^d$. We still take the periodic graph $\Gamma$ with its associated displacement graph $(\Gamma/\Z^d,\delta)$ fixed. As before, $||\cdot||$ will be a fixed but arbitrary norm on $\R^d$. The constant 
$$
C\equiv \max_{e\in E}||\delta(e)||
$$
will be of some use. $|V|$ will always stand for $|V_{\Gamma/\Z^d}|$.

\begin{lem}
\label{finpath} For every path $p$ in $\Gamma/\Z^d$, there are cycles $c_1,\ldots,c_k$ in $\Gamma/\Z^d$ such that
$$
\left|\left|\delta(p)-\sum_{i=1}^k \delta(c_i)\right|\right| < C\, |V| \qquad\textit{and}\qquad 0\leq |p|-\sum_{i=1}^k |c_i|\leq |V| \:.
$$
\end{lem}

\begin{proof}
For $|p|<|V|$, there is nothing to prove since one can just take $k=0$, i.e.~the sum over cycles to be empty.
For $|p|\geq|V|$, we use induction on $|p|$. By Lemma~\ref{fingraph}\ref{pathcycle}, the path $p$ can be written in the form
$$
p=p_0 c p_1
$$
where $c$ is a cycle, so that the paths $p_0$ and $p_1$ can be composed to $p'=p_0 p_1$. Since $|c|\geq 1$, we conclude $|p'|=|p_0|+|p_1|<|p|$, so that an application of the induction assumption to $p'$ gives cycles $c_2,\ldots,c_k$ with
$$
\left|\left|\delta(p')-\sum_{i=2}^k \delta(c_i)\right|\right| < C\, |V| \qquad\textit{and}\qquad 0\leq |p'|-\sum_{i=2}^k |c_i|\leq |V| \:.
$$
The conclusion follows by setting $c_1=c$ and observing $\delta(p)=\delta(c_1)+\delta(p')$ and $|p|=|c_1|+|p'|$.
\end{proof}

\begin{thm}[Main theorem]
\label{mainthm}
Let $\Gamma$ be such that $\Gamma/\Z^d$ is strongly connected. Then the set of velocities of $\Gamma$ coincides with the convex hull 
\beq
P_\Gamma \equiv \mathrm{conv}\left\{u_f\:|\: u_f \textrm{ basic velocity in }\Gamma/\Z^d\right\} \:.
\eeq 
In particular, $P_\Gamma$ is a rational polytope in $\R^d$, the \emph{velocity polytope} of $\Gamma$.
\end{thm}

We now give an outline of the proof before diving into the details. The idea is as follows: if $c_1$ and $c_2$ are cycles in $\Gamma/\Z^d$ with the same initial vertex, then the closed path $c_1c_2$ has a velocity which is a convex combination of the basic velocities associated to $c_1$ and $c_2$. An analogous statement holds for longer combinations of cycles. 

So in order to show that the velocity of a trajectory is always a convex combination of basic velocities, one can apply Lemma~\ref{finpath} in order to decompose the trajectory into cycles, noting that the right-hand side becomes irrelevant in the limit.

Conversely, for every convex combination of basic velocities one needs to construct a trajectory which has this velocity. By choosing the number of times that each cycle appears in a closed path, one can adjust the coefficients of the convex combination which corresponds to the velocity of (the lift of) that closed path. Therefore, one can try to combine the cycles such that they appear in the trajectory with the appropriate frequencies, while also inserting some auxiliary paths which connect between cycles with different starting vertices.

We now implement this strategy in detail.

\begin{proof}[Proof of Theorem~\ref{mainthm}]
Working with the displacement graph $(\Gamma/\Z^d,\delta)$ instead of the periodic graph $\Gamma$, we begin by showing that any velocity lies in the convex hull of the basic velocities. We first claim that for any trajectory $f$ and any $n\in\N$, the quotient
$$
w_n = \frac{\sum_{i=1}^n \delta(f_i)}{n}
$$
has the property that there is a vector $u$ in the convex hull of the basic velocity vectors such that
\beq
\label{claim}
||w_n-u|| < \frac{2\,|V|\,C}{n}
\eeq
To this end, we first approximate the path $f_1\ldots f_n$ as in Lemma~\ref{finpath} by cycles $c_1,\ldots,c_k$, so that
\beq
\label{twoineqs}
\left|\left|\sum_{i=1}^n \delta(f_i) - \sum_{j=1}^k \delta(c_j)\right|\right| < C\,|V| \qquad\textrm{and}\qquad 0\leq n-\sum_{j=1}^k |c_j| \leq |V|\:.
\eeq
The first of these two inequalities can be rewritten as
\beq
\label{wnineq}
\left|\left|w_n - \sum_{j} \frac{|c_j|}{n} \cdot \frac{\delta(c_j)}{|c_j|} \right|\right| < \frac{C\,|V|}{n}
\eeq
Since the fractions $\delta(c_j)/|c_j|$ are basic velocities, the expression
$$
u\equiv \sum_j \frac{|c_j|}{\sum_l |c_l|} \cdot \frac{\delta(c_j)}{|c_j|} 
$$
is a convex combination of basic velocities; by the second inequality of\eq{twoineqs}, this approximates well the term $\sum_j\frac{|c_j|}{n}\cdot\frac{\delta(c_j)}{|c_j|}$ appearing in\eq{wnineq}. More precisely,
$$
\left|\left| w_n - u \right|\right| = \left|\left| w_n - \frac{\sum_l|c_l|}{n}u - \left(1-\frac{\sum_l|c_l|}{n}\right)u \right|\right| \stackrel{(\ref{wnineq})}{<} \frac{C\,|V|}{n} + \left|\left|\left(1-\frac{\sum_l|c_l|}{n}\right)u\right|\right|
$$
Since $||u||\leq C$ and $0\leq n-\sum_l |c_l|\leq |V|$, we can also bound the second term on the right-hand side by $C|V|/n$, which proves the claim\eq{claim}.

It follows from\eq{claim} that when the trajectory has a well-defined velocity $\lim_n w_n$, then the distance from this velocity to the convex hull of basic velocities is smaller than $2|V|C/n$ for any $n$. Since that convex hull is a polytope and therefore closed, the limit $\lim_n w_n$ is itself in the convex hull of basic velocities.

Conversely, it has to be shown that any convex combination of basic velocities
\beq
\label{convvel}
\sum_{i=1}^r \lambda_i \frac{\delta(c_i)}{|c_i|}
\eeq
for weights $\lambda_1,\ldots,\lambda_r\geq 0$ with $\sum_i\lambda_i=1$ and cycles $c_i$ can be realized by a trajectory. In order to construct such a trajectory, let us choose natural numbers $\alpha_{ik}$ giving rational approximations to the numbers $\lambda_i/|c_i|$ as
\beq
\label{alpha}
\alpha_{ik}\equiv\left\lfloor k\cdot\frac{\lambda_i}{|c_i|}\right\rfloor\qquad\textrm{so that} \qquad \left|\frac{\alpha_{ik}}{k} - \frac{\lambda_i}{|c_i|} \right| \lt \frac{1}{k} \quad \forall i,k\in\N \:.
\eeq
Let us also choose paths $p_1,\ldots,p_r$ in $\Gamma/\Z^d$ such that $|p_i|<|V|$ and $p_i$ connects $t(c_i)$ to $s(c_{i+1})$ (with $c_{r+1}=c_1$, so that $p_r$ connects $t(c_r)$ to $s(c_1)$). Such a choice of paths exists due to the assumption of strong connectivity. Then the building blocks of the trajectory are going to be the paths
\beq
\label{defqk}
q_k \equiv \underbrace{c_1\ldots c_1}_{\alpha_{1k}\:\textrm{times}} p_1 \underbrace{c_2\ldots c_2}_{\alpha_{2k}\:\textrm{times}} p_2 \ldots \underbrace{c_r\ldots c_r}_{\alpha_{rk}\:\textrm{times}} p_r \:.
\eeq
Each $q_k$ is a closed path in $\Gamma/\Z^d$ in the sense that $s(q_k)=t(q_k)$. Its length can be estimated as
\beq
\label{qest}
|q_k|=\sum_{i=1}^r \alpha_{ik} |c_i| + \sum_{i=1}^r |p_i| = \sum_{i=1}^r\left( k\cdot\frac{\lambda_i}{|c_i|}\cdot|c_i| + O(1)\right) = k + O(1) \:,
\eeq
where $O(1)$ refers to a term which does not depend on $k$. This implies the rough estimate $||\delta(q_k)|| \leq O(k)$, which we record for future use.

The trajectory $f$ is defined to be (the lift to $\Gamma$) of the infinite path
\beq
\label{deff}
f\::\: q_1 q_2 q_2 q_3 q_3 q_3 \ldots \underbrace{q_k\ldots q_k}_{k\:\textrm{times}} \ldots \:.
\eeq
It needs to be shown that this trajectory has a velocity which equals~(\ref{convvel}).
We consider this trajectory up to a timestep $n$, i.e. the path $f_1\ldots f_n$. We choose $k$ (as a function of $n$) such that the path $f_1\ldots f_n$ has already concluded the whole $q_k$-segment from\eq{deff}, but not yet the whole $q_{k+1}$-segment. Then due to\eq{qest}, the $q_{k+1}$-segment will contribute to the time $n$ and the displacement $\delta$ by at most $O(k^2)$. Hence the total displacement up to $n$ timesteps is given by
$$
\sum_{i=1}^n \delta(f_i) = \sum_{m=1}^k m\cdot \delta(q_m) + O(k^2) \stackrel{(\ref{defqk})}{=} \sum_{m=1}^k \sum_{i=1}^r m\left(\alpha_{im} \delta(c_i) + \delta(p_i) \right) + O(k^2)
$$
Using condition~(\ref{alpha}) in the form $\left|\alpha_{im} - \lambda_i \frac{m}{|c_i|}\right|\leq O(1)$ evaluates this to
$$
\sum_{i=1}^n \delta(f_i) = \sum_{m=1}^k \sum_{i=1}^r m\left(\frac{\lambda_i m}{|c_i|} \delta(c_i) + O(1) \right) + O(k^2) = \frac{k^3}{3}\sum_{i=1}^r \lambda_i \frac{\delta(c_i)}{|c_i|}+O(k^2)
$$
On the other hand, the number $n$ of edges traversed, which equals the time elapsed, can be evaluated very similarly as
$$
n = \sum_{m=1}^k m\cdot |q_m| + O(k^2) = \sum_{m=1}^k \sum_{i=1}^r m\left(\alpha_{im} |c_i| + |p_i| \right) + O(k^2)
$$
$$
=\sum_{m=1}^k \sum_{i=1}^r m\left(\lambda_i m + O(1) \right) + O(k^2) = \frac{k^3}{3}\sum_{i=1}^r \lambda_i + O(k^2) = \frac{k^3}{3} + O(k^2) \:.
$$
Hence the velocity of the trajectory is given by
$$
\lim_{n\to \infty} \frac{\sum_{i=1}^n \delta(f_i)}{n} =
\lim_{k\to \infty} \frac{ \frac{k^3}{3}\sum_{i=1}^r \lambda_i \frac{\delta(c_i)}{|c_i|}+O(k^2) } { \frac{k^3}{3} + O(k^2) } = \sum_{i=1}^r \lambda_i \frac{\delta(c_i)}{|c_i|} \:,
$$
as desired.
\end{proof}

\begin{ex}
\label{hc3}
For the periodic graph of Figure~\ref{honeycomb} and~\ref{hcrect}, the velocity polytope is the hexagon in Figure~\ref{hcvp}.
\end{ex}

We now consider what happens when the connectedness assumption of Theorem~\ref{mainthm} is dropped.

\begin{prop}
\label{nonconn}
Let $\Gamma$ be any periodic graph. Let $\Gamma_1/\Z^d,\ldots,\Gamma_c/\Z^d$ be the strongly connected components of $\Gamma/\Z^d$ which have preimages $\Gamma_1,\ldots,\Gamma_c\subseteq\Gamma$. In this case, the set of velocities of $\Gamma$ is the union of polytopes
\beq
\label{fup}
P_\Gamma = P_{\Gamma_1}\cup \ldots \cup P_{\Gamma_c} \:.
\eeq
\end{prop}

\begin{proof}
The $\Gamma_i$ are defined as the preimages under $\phi:\Gamma\to\Gamma/\Z^d$ of the strongly connected components of $\Gamma/\Z^d$. Then for every trajectory $(f_n)_{n\in\N}$, there is $n_0\in\N$ such that all $f_n$ for $n\geq n_0$ lie in the same $\Gamma_i$. The velocity of $f$ therefore lies in the corresponding velocity polytope $P_{\Gamma_i}$. This proves the ``$\subseteq$'' inclusion of\eq{fup}. The ``$\supseteq$'' inclusion is clear since any trajectory in some $\Gamma_i$ is also a trajectory in $\Gamma$.
\end{proof}

\section{The large-scale geometry of a periodic graph}

We now relate the velocity polytope $P_{\Gamma}$ to the large-scale geometry of $\Gamma$. What we mean by this is the following. There is a natural notion of distance between two vertices defined to be the length of the shortest path connecting the two vertices. This defines a metric $d(\cdot,\cdot)$ on $\Gamma$ invariant under the action of $\Z^d$. Given any $x\in\Z^d$ and any vertex $v\in V_\Gamma$, we can now define the $\Gamma$-norm $||x||_{\Gamma}$ to be given by
\beq
\label{defnorm}
||x||_{\Gamma} \equiv \lim_{n\ra\infty} \frac{d(v+nx,v)}{n}
\eeq
The existence of the limit is guaranteed by the triangle inequality and Fekete's Lemma~\cite{Fekete}. Taking the limit instead of defining $||x||_\Gamma$ to be $d(v,v+x)$ itself is necessary in order guarantee that $||x||_\Gamma$ does not depend on $v$. This well-definedness of\eq{defnorm} easily follows from arguments very similar to those made in the proof of Lemma~\ref{iotaind}. Moreover, it is simple to show that $||mx||_\Gamma=m\,||x||_\Gamma$ for any $m\in\N$. The norm $||x||_\Gamma$ can be interpreted as follows: for any $v\in\Gamma$, one needs to traverse $n\cdot ||x||_\Gamma+O(1)$ edges in order to get from $v$ to $v+xn$.

From the triangle inequality $||x+y||_\Gamma\leq ||x||_\Gamma+||y||_\Gamma$ and $||mx||_\Gamma=m\,||x||_\Gamma$ for $m\in\N$, one deduces that $||\cdot||_\Gamma$ extends to a unique norm on $\R^d$.

\begin{prop}
\label{lsgeometry}
Let $\Gamma$ be strongly connected. Then $P_\Gamma$ is the unit ball of $||\cdot||_\Gamma$.
\end{prop}

\begin{proof}
We first show that $P_\Gamma$ is contained in the unit ball. To this end, it is enough to prove $||u_c||_\Gamma\leq 1$ for a basic velocity $u_c$ associated to a cycle $c=e_1\ldots e_k$ in $\Gamma/\Z^d$. Choosing any $v\in\phi^{-1}(s(e_1))$, the cycle $c$ lifts to a path in $\Gamma$ from $v$ to $v+\delta(c)$. Since this path has length $k$, we get
$$
||u_c||_\Gamma = \left|\left| \frac{\delta(c)}{k} \right|\right|_\Gamma = \frac{1}{k} \lim_{n\ra\infty} \frac{d(v,v+n\delta(c))}{n} \leq \frac{1}{k} d(v,v+\delta(c)) \:.
$$
There is a path from $v$ to $v+\delta(c)$ of length $k$, which implies $d(v,v+\delta(c))\leq k$. This results in the desired inequality $||u_c||_\Gamma\leq 1$.

Conversely, it has to be shown that $||x||_\Gamma\leq k$ for $x\in\Z^d$ and $k\in\N$ implies that $\frac{x}{k}\in P_\Gamma$. Fix $\eps>0$. By assumption, there are $v\in\Gamma$ and $n\in\N$ such that
$$
k-\eps \leq \frac{d(v,v+nx)}{n} \leq k+\eps \:.
$$
This means that there exists a path $p$ from $v$ to $v+nx$ of length between $nk-n\eps$ and $nk+n\eps$. Since $t(p)$ is a translate of $s(p)$, the path $p$ can be concatenated with its own translates in order to form a trajectory which periodically traverses translates of $p$. Since $\delta(p)=nx$, the velocity of this trajectory is given by $nx/|p|$, so that
$$
\frac{x}{k} \cdot \frac{kn}{|p|} \in P_\Gamma \:.
$$
In terms of an arbitrary norm $||\cdot||$ on $\R^d$, the distance of $\frac{x}{k}$ to $P_\Gamma$ can therefore be bounded by
$$
\left|\left|\frac{x}{k} - \frac{x}{k} \cdot \frac{kn}{|p|} \right|\right| = \frac{||x||}{k} \cdot \frac{1}{|p|} \cdot \big||p|-nk \Big| \leq \frac{||x||}{k} \cdot \frac{\eps}{k-\eps} \:.
$$
Since this vanishes as $\eps\ra 0$, we conclude that $\frac{x}{k}\in P_\Gamma$ from closedness of $P_\Gamma$.
\end{proof}

\section{Properties of velocity polytopes} We now study some basic properties of velocity polytopes. 

\begin{prop}
\label{converse}
Every rational polytope arises as the velocity polytope $P_\Gamma$ of an appropriate periodic graph $\Gamma$.
\end{prop}

\begin{proof}
Let $P\subseteq\R^d$ be a non-empty polytope with rational vertices $w_1,\ldots,w_m\in\Q^d$. Let $\gamma$ be the least common multiple of the denominators of the $w_i$, so that $\gamma w_i \in\Z^d$ for all $i$. Then $P$ can be realized as a velocity polytope as follows: let us construct a displacement graph on vertices $u_1,\ldots,u_{\gamma}$ such that there is a single edge with displacement $\delta=0$ from the vertex $u_j$ to the vertex $u_{j+1}$ for each $j=1,\ldots,\gamma-1$, and additional edges $e_1,\ldots,e_m$ from $u_\gamma$ to $u_1$ with displacements
$$
\delta(e_i)\equiv w_i \:.
$$
This defines a strongly connected displacement graph. Its basic velocities are precisely all the $\frac{w_i}{\gamma}$. Therefore, the associated periodic graph has $P$ as its velocity polytope.
\end{proof}

In this paper, we have been considering the general case of directed graphs; since every undirected graph can be made into a directed graph by replacing an undirected edge by a pair of oppositely oriented directed edges, the theory also applies to undirected graphs. In an undirected graph, every path can be reversed, which reverses the sign of its velocity. Therefore, it should not be surprising that the following holds:

\begin{prop}
When $\Gamma$ is undirected, then $P_{\Gamma}$ is symmetric around the origin.
\end{prop}

\begin{proof}
The edges of $\Gamma$ come in parallel pairs of opposite orientation. Therefore, the edges of $\Gamma/\Z^d$ also come in parallel pairs with opposite orientation and displacement of the opposite sign. Hence for every cycle $c$ in $\Gamma/\Z^d$, there is a cycle $c'$ which corresponds to traversing $c$ backwards by using all the ``partner'' edges. Therefore, if $u_c$ is a basic velocity, then so is $-u_c$. Now the statement follows from Theorem~\ref{mainthm} and Proposition~\ref{nonconn}.
\end{proof}

\begin{prop}[Connectedness]
\label{connectedness}
If $\Gamma$ itself is strongly connected, then $P_{\Gamma}\subseteq\R^d$ is full-dimensional and $0\in P_{\Gamma}$ is an interior point.
\end{prop}

\begin{proof}
This is best proven without appealing to Theorem~\ref{mainthm}. Let $\epsilon_1,\ldots,\epsilon_d\in\Z^d$ be the standard unit vectors, and $v\in\Gamma$ a fixed starting vertex. Then by the assumption of strong connectedness of $\Gamma$, there is a path $p_i^+$ in $\Gamma$ from $v$ to $v+\epsilon_i$. In $\Gamma/\Z^d$, this is a closed path of displacement $\epsilon_i$ and velocity $\frac{\epsilon_i}{|p_i^+|}$. We can use translates of $p_i^+$ to connect $v+n\epsilon_i$ to $v+(n+1)\epsilon_i$ for any $n\in\N$. Sequentially traversing these translates of $p_i^+$ defines a trajectory with velocity $\frac{\epsilon_i}{|p_i^+|}$. Similarly, there exists a path $p_i^-$ from $v$ to $v-\epsilon_i$, which gives rise to a trajectory with velocity $-\frac{\epsilon_i}{|p_i^-|}$. The convex hull of these $2d$ velocities is a subset of $P_\Gamma$. By construction, this subset is full-dimensional and includes the origin as an interior point, and so the same also holds for $P_{\Gamma}$.
\end{proof}

The converse is not true: for example, in dimension $d=1$ we can take $\Gamma/\Z^d$ to consist of a single vertex with a loop $e_+$ of displacement $\delta(e_+)=2$ and a loop $e_-$ of displacement $\delta(e_-)=-2$, which makes $\Gamma/\Z^d$ strongly connected and gives $P_{\Gamma}=[-2,2]$, although the associated periodic graph $\Gamma$ has two connected components.

So far, we have been talking about the velocity polytope of a single periodic graph $\Gamma$. But given two periodic graphs $\Gamma$ and $\Gamma'$, how do their velocity polytopes relate? In order to find some relation between $P_\Gamma$ and $P_{\Gamma'}$, one needs to assume a relation between $\Gamma$ and $\Gamma'$. One such notion is that of a \emph{morphism} $h:\Gamma\to\Gamma'$ between periodic graphs $\Gamma$ and $\Gamma'$ of the same dimension $d$, by which we mean maps
$$
h_V \: : \: V_\Gamma \longrightarrow V_{\Gamma'} \:,\qquad h_E \: : \: E_\Gamma \longrightarrow E_{\Gamma'} \:,
$$
which are compatible with the graph structures,
$$
s_{\Gamma'} \circ h_E = h_V \circ s_\Gamma \:,\qquad t_{\Gamma'} \circ h_E = h_V \circ t_\Gamma \:,
$$
and the $\Z^d$-action,
\beq
\label{morphaction}
h_V(v+x) = h_V(v) + x \:,\quad h_E(e+x) = h_E(e) + x \:.
\eeq
A morphism of periodic graphs induces a simple relationship between the velocity polytopes:

\begin{prop}[Functoriality]
Let $h:\Gamma\to\Gamma'$ be a morphism of periodic graphs. Then
$$
P_{\Gamma} \subseteq P_{\Gamma'} \:.
$$
\end{prop}

\begin{proof}
We consider how $h$ operates on the quotient graphs $\Gamma/\Z^d$ and $\Gamma'/\Z^d$. Since by\eq{morphaction} the assignment $h_V$ maps $\Z^d$-orbits to $\Z^d$-orbits, we get induced maps $\widehat{h}_V:V_{\Gamma/\Z^d}\to V_{\Gamma'/\Z^d}$ and $\widehat{h}_E:E_{\Gamma/\Z^d}\to E_{\Gamma'/\Z^d}$. Moreover, the orbit representatives $\iota:V_{\Gamma/Z^d}\ra V_\Gamma$ and $\iota':V_{\Gamma'/\Z^d}\ra V_{\Gamma'}$ can be chosen such that $h_V$ maps representatives to representatives in the sense that $h_V\circ\iota = \iota'\circ \widehat{h}_V$. In total, there is a commutative diagram
$$
\xymatrix{ \Gamma/\Z^d \ar[r]^{\widehat{h}_V} \ar[d]_\iota & \Gamma'/\Z^d \ar[d]^{\iota'} \\
\Gamma \ar[r]^{h_V} \ar[d]_\phi & \Gamma \ar[d]^{\phi'} \\
\Gamma/\Z^d \ar[r]^{\widehat{h}_V} & \Gamma/\Z^d }
$$
such that the vertical compositions are identities. Then it follows that the induced map $\widehat{h}_E:E_{\Gamma/\Z^d}\to E_{\Gamma'/\Z^d}$ is compatible with the displacements, $\delta'(\widehat{h}(e)) = \delta(e)$ for all edges $e\in E_{\Gamma/\Z^d}$.

To prove the assertion, we now show that every velocity of $\Gamma$, associated to a trajectory $(f_n)_{n\in\N}$, is also a velocity in $\Gamma'$. But this follows from
$$
\lim_{n\ra\infty}\frac{\sum_{k=1}^n \delta(f_k)}{n} = \frac{\sum_{k=1}^n \delta'(\widehat{h}(f_k))}{n} \:,
$$
so that $\left(h(f_n)\right)_{n\in\N}$ is a trajectory in $\Gamma'$ with the same velocity.
\end{proof}

This result can be applied for example as follows:

\begin{cor}
\label{morphapp}
If $P_{\Gamma}\not\subseteq P_{\Gamma'}$, then there is no morphism $h:\Gamma\ra\Gamma'$.
\end{cor}

\section{A no-go theorem for digital physics}
\label{physics}

Many recent proposals in fundamental physics revolve around the idea that space is, in some sense, discrete~\cite{A1,A2,BHmodel,Qgraphity,tHooft}. This is motivated by the conviction that the unification of quantum theory with general relativity will require the introduction of a minimal length scale~\cite{Garay}. Building theories of physics in which space is fundamentally discrete is sometimes also known as \emph{digital physics}~\cite{Fredkin}. If the idea of digital physics is correct, then one needs to ask: what will be the observational consequences of this fundamental discreteness? The \emph{tile argument} due to Hermann Weyl~\cite[p.~43]{Weyl} suggests that such fundamental discreteness can manifest itself on the macroscopic scale:

\begin{quote}
``If a square is built up of miniature tiles, then there are as many tiles along the diagonal as there are along the side; thus the diagonal should be equal in length to the side.''
\end{quote}

We refer to~\cite{VB} for an extensive discussion and references on the tile argument. Among its weaknesses are the lack of a precise definition of ``length'' (length of what?) and that it only applies to a square lattice embedded in Euclidean space. In particular, the tile argument does not answer the question whether there could be discrete models of physics for which the continuum limit corresponds to ordinary (non-relativistic) physics in Euclidean space.

We now explain the sense in which our Theorem~\ref{mainthm} can be interpreted as answering this question in the negative. The kind of discrete models that we consider are classical point particles on a periodic graph $\Gamma$. While more general frameworks for discrete models of physics are certainly conceivable, our present results are limited to this case. A classical point particle on $\Gamma$ has a \emph{trajectory} as in Definition~\ref{deftraj}. (It is not a strong restriction to assume that the particle hops along precisely one edge per timestep. More complicated models in which e.g.~different edges behave in different ways will often be equivalent to a different model satisfying our assumptions. For example, if there is a certain kind of edge that takes two timesteps to traverse, just subdivide each such edge into two ``ordinary'' edges.)

Now, as Remark~\ref{physics} shows, the \emph{velocity} of a trajectory in the sense of~(\ref{vel}) corresponds to the usual intuitive notion of velocity, possibly up to an affine transformation. We would like to emphasize that our definition~(\ref{vel}) does neither require an embedding of the periodic graph $\Gamma$ into Euclidean space, nor does it presuppose a notion of distance on $\Gamma$. Both of these properties are features that one should expect a reasonable self-contained discrete model of physics to have. More concretely, although we can talk about velocities, we cannot talk about the \emph{speed} of a trajectory ($=$magnitude of its velocity), due to the absence of a meaningful notion of distance or length.

If a discrete model is to recover the usual Euclidean space in the continuum limit, then it also needs to recover its symmetries. In particular, the set of possible velocities of particles should be a ball of a certain radius, in accordance with the perceived \emph{isotropy} of space: all directions look the same, and in particular the maximal speed of a particle does not depend on the direction of its velocity. In our framework, this corresponds to the requirement that the set of possible velocities should be \emph{ellipsoidal} in shape, so that an appropriate linear transformation maps the ellipsoid into a ball. Differently phrased, the set of velocities itself should determine a Euclidean metric such that the set turns into a ball when using that dynamically determined metric. By Proposition~\ref{lsgeometry}, this (hypothetically) Euclidean metric should coincide with the metric induced from geodesic distance on $\Gamma$.

However, our Theorem~\ref{mainthm} (more generally, Proposition~\ref{nonconn}) implies that the set of velocities of particles on a periodic graph can never be ellipsoidal in shape. Alternatively speaking, the large-scale geometry of a periodic graph is never Euclidean with respect to any metric. Therefore, we arrive at the following \emph{no-go theorem}:

\begin{thm}
\label{nogo}
There is no periodic graph $\Gamma$ for which the set of macroscopic velocities achievable by a classical point particle hopping along the edges is isotropic.
\end{thm}

In particular, the anisotropies present in $\Gamma$ will never be suppressed in the continuum limit, as one might na{\"i}vely expect. On the other hand, our Proposition~\ref{converse} shows that the remaining macroscopic anisotropies can be made as small as desired. The price that has to be paid consists in having to work with graphs with many edges per unit cell.

Our present analysis is limited to \emph{classical} point particles, which is not very realistic since actual particles are governed by \emph{quantum} mechanics. We strongly suspect that a statement analogous to Theorem~\ref{nogo} will also hold for quantum particles; for example, anisotropies are visible for quantum particles on hexagonal lattices at large enough momenta, an effect known as \emph{trigonal warping}~\cite{tw}.

\begin{prob}
State and prove a version of Theorem~\ref{nogo} for a quantum-mechanical particle on a periodic graph.
\end{prob}

We suspect that this is closely related to the behavior of electrons in periodic potentials, one of the main topics studied in solid-state physics.

\bibliographystyle{plain}
\bibliography{velocity_polytopes}

\end{document}